\documentclass[a4paper]{llncs}
\usepackage{amsfonts,amssymb,amsmath,latexsym,ae,aecompl}
\usepackage{float}
\usepackage{epsfig}
\usepackage{enumerate}
\usepackage{graphicx}
\usepackage{fourier}
\usepackage{algorithmic}
\usepackage{caption}
\usepackage{graphicx}
\usepackage{subfig}
\usepackage{gensymb}
\usepackage{framed,color}
\usepackage[ruled,vlined,linesnumbered]{algorithm2e}

\title{Self-Stabilizing Robots without Location Information}

\title{Convergence of Even Simpler Robots\\without Location Information
}

\author{Debasish Pattanayak$^1$ \quad Kaushik Mondal$^1$ \\
Partha Sarathi Mandal$^1$ \quad Stefan Schmid$^{2,3}$\thanks{Trip to IIT Guwahati and Research were funded by Global Initiative of Academic Networks (GIAN), an initiative by Govt. of India for Higher Education.}}
\institute{\small $^1$Indian Institute of Technology Guwahati, India\\
$^2$Aalborg University, Denmark \quad $^3$TU Berlin, Germany
\\
}
\begin{document}
\maketitle
\begin{abstract}
The design of distributed gathering and convergence algorithms
for tiny robots has recently received much attention. In particular, it has been shown that
convergence problems can even be solved for very weak, \emph{oblivious}
robots: robots which cannot maintain state from one round to the next.
The oblivious robot model is hence attractive from a self-stabilization
perspective, where state is subject to adversarial manipulation.
However, to the best of our knowledge, all existing robot convergence protocols rely
on the assumption that robots, despite being ``weak'', can measure distances.

 We in this paper initiate the study of convergence protocols for even simpler robots,
 called \emph{monoculus robots}:
 robots which cannot measure distances. In particular, we introduce two
 natural models which relax the assumptions
 on the robots' cognitive capabilities: (1) a Locality Detection ($\mathcal{LD}$) model in which a robot
 can only detect whether another robot is closer than a given constant distance
 or not, (2)
 an Orthogonal Line Agreement ($\mathcal{OLA}$) model in which robots only agree on a pair of orthogonal lines
 (say North-South and West-East, but without knowing which is which).

 The problem turns out to be non-trivial,
 and simple median and angle bisection strategies can easily increase the distances among
 robots (e.g., the area of the enclosing convex hull) over time.
 Our main contribution are deterministic self-stabilizing convergence algorithms for
 these two models, together with a complexity analysis.
 We also show that in some sense,
 the assumptions made in our models are minimal: by relaxing the assumptions on the \textit{monoculus robots}
 further, we run into impossibility results.
\end{abstract}

\sloppy

\section{Introduction}

\subsection{The Context: Tiny Robots}

In the recent years, there has been a wide interest in
the cooperative behavior of tiny robots.
In particular, many distributed coordination protocols have been devised
for a wide range of models and for a wide range of problems,
like convergence, gathering, pattern formation,
flocking, etc.
At the same time, researchers have also started characterizing
the scenarios in which such problems cannot be solved, deriving
impossibility results.

\subsection{Our Motivation: Even Simpler Robots}

An interesting question regards the minimal cognitive
capabilities that such tiny robots need to have for completing a particular task.
In particular, researchers have initiated the study of ``weak robots''\cite{FlocchiniPSW99}.
Weak robots are \textit{anonymous} (they do not have any identifier),
\textit{autonomous} (they work independently),
\textit{homogeneous} (they behave the same in the same situation),
and \textit{silent} (they also do not communicate with each other).

Weak robots are usually assumed to have their own local view, represented as a Cartesian
coordinate system with origin and unit length and axes. The orientation of axes, or the \textit{chirality} (relative order of the orientation of axes or handedness),
is not common among the robots.
The robots move in a sequence of three consecutive actions, \textit{Look-Compute-Move}:
they observe the positions of other robots in their local coordinate system and the observation step returns a set of points to the observing robot.
The robots cannot distinguish if there are multiple robots at the same position,
i.e., they do not have the capability of \textit{multiplicity detection}.
Importantly, the robots are \textit{oblivious} and cannot maintain state between rounds
(essentially moving steps).
The computation they perform are always based on the data they have collected in the \emph{current} observation step; in the next round they again collect the data.
Such weak robots are therefore interesting
from a self-stabilizing perspective: as robots do not rely on memory,
an adversary cannot manipulate the memory either.
Indeed, researchers have demonstrated that weak robots
are sufficient to solve a wide range of problems.

We in this paper aim to relax the assumptions on the tiny robots
further. In particular, to the best of our knowledge, all prior literature
assumes that robots can observe the positions of other robots in their local view. This enables them to calculate the distance between any pair of robots.
This seems to be a very strong assumption, and accordingly,
we in this paper initiate the study of even weaker robots which
cannot locate other robots positions in their local view, preventing them from measuring distances. We define these kind of robots as \textit{monoculus robots}.

In particular, we
initiate to explore two naturally weaker models for monoculus robots with less cognitive capabilities
\begin{enumerate}
\item \emph{Locality Detection} ($\mathcal{LD}$): The robots can distinguish whether a neighbor robot is at a distance more than a predefined value $c$ or not.
\item \emph{Orthogonal Line Agreement} ($\mathcal{OLA}$): The robots agree on a pair of orthogonal lines (but not necessarily the orientation of the lines).
\end{enumerate}

\subsection{The Challenge: Convergence}

We focus on the fundamental convergence problem for monoculus robots
and show that the problem is already non-trivial in this setting.

In particular, many naive strategies lead to non-monotonic
behaviors. For example, strategies where boundary robots
(robots located on the convex hull) move toward the
``median'' robot they see, may actually \emph{increase}
the area of the convex hull in the next round,
counteracting convergence as shown in Fig.~\ref{fig:anglebisector}$(a)$.
A similar counterexample exists for a strategy where
robots move in the direction of the angle bisector as shown in Fig.~\ref{fig:anglebisector} ($b$).
\noindent
\begin{figure}[H]\centering
\includegraphics[width=\linewidth]{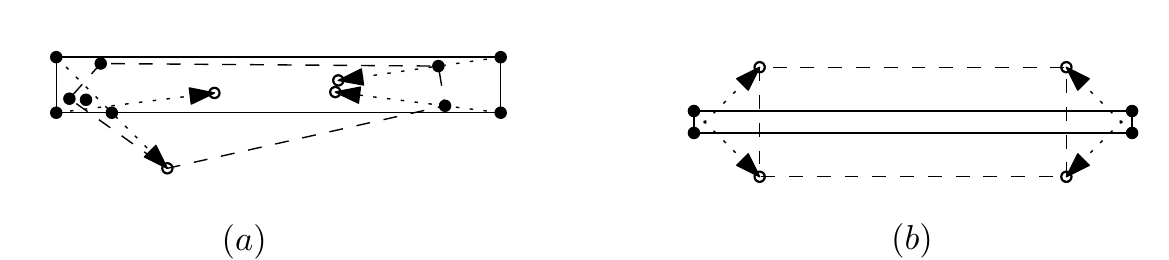}
\caption{ The 4 boundary robots are moving $(a)$ towards the median robot $(b)$ along the angle bisector. The discs are the old positions and circles are the new positions. The old convex hull is drawn in solid line, the new convex hull is dashed. The arrows denote the direction of moving.  }\label{fig:anglebisector}
\end{figure}
But not only enforcing convex hull invariants is challenging,
also the fact that visibility is restricted and we cannot detect
multiplicity:
We in this paper assume that robots are not transparent, and
accordingly, a robot does not see whether and how many robots
may be hidden behind a visible robot. As robots are also not able
to perform multiplicity detection (i.e., determine how many robots
are collocated at a certain point), strategies such as ``move toward
the center of gravity'' (the direction in which most robots are located),
are not possible.

\subsection{Our Contributions}

This paper studies distributed convergence problems
for anonymous, autonomous, oblivious, non-transparent, monoculus, point robots under a most general asynchronous scheduling model
and makes the following contributions.
\begin{enumerate}
\item We initiate the study of a new kind of robot,
 the \emph{monoculus robot} which cannot measure distances.
 The robot comes in two natural flavors, and we introduce the
 Locality Detection ($\mathcal{LD}$)
and the Orthogonal Line Agreement ($\mathcal{OLA}$) model accordingly.
\item We present and formally analyze deterministic and self-stabilizing distributed convergence algorithms for both
$\mathcal{LD}$ and $\mathcal{OLA}$.
\item We show our assumptions in $\mathcal{LD}$ and $\mathcal{OLA}$ are minimal in the sense that
robot convergence is not possible for monoculus robots without any additional capability.
\item We report on the performance of our algorithms through simulation.
\item We show that our approach can be generalized to higher dimensions and, with a small extension, supports termination.
\end{enumerate}

\subsection{Paper Organization}

The remainder of this paper is organized as follows.
Section~\ref{sec:preli} introduces the necessary
background and preliminaries.
Section~\ref{sec:algo} introduces two algorithms for convergence.
Section~\ref{sec:impossibilty} presents an impossibility result which shows the minimality of our assumptions.
We report on simulation results in
Section~\ref{sec:simul} and discuss extensions in
Section~\ref{sec:discussion}.
In Section~\ref{sec:relwork}, we review related work,
before we conclude in Section~\ref{sec:conclusion}.

\section{Preliminaries}\label{sec:preli}

\subsection{Model}

We consider anonymous, autonomous, homogeneous, oblivious, non-transparent robots with unlimited visibility, unless the view is obstructed by another robot:
Since the robots are non-transparent, any robot can see at most one robot in any direction.
As usual, the robots in each round execute a sequence of \textit{Look-Compute-Move} steps:
First, the robot observes other robots (\textit{Look} step); second, on the basis of the observed information, it executes an algorithm
which computes a direction which the robot must move towards
(\textit{Compute} step);
the robot then moves in this direction (\textit{Move} step),
for a fixed distance $b$ (the step size).
The robots are silent, cannot detect multiplicity points, and can
pass over each other (no collision occurs).

In this paper, we introduce monoculus robots:
\begin{definition}(Monoculus Robot)
A robot is called \emph{monoculus} if it is anonymous, autonomous, oblivious, homogeneous, and silent. We assume the robot is a non-transparent point
robot, has unlimited visibility, and can neither determine the position of other robots nor detect multiplicty.
\end{definition}

We consider the most general CORDA or ASYNC scheduling model
known from weak robots~\cite{FlocchiniPSW99} as well as
the ATOM or Semi-Synchronous (SSYNC) model~\cite{suzu}.
These models define the activation schedule of the robots:
the SSYNC model considers instantaneous computation and movement, i.e.,
the robots cannot observe other robots in motion,
while in the ASYNC model any robot can look at any time. In SSYNC the time is divided into global rounds and a subset of the robots are activated in each round
which finish their \textit{Look-Compute-Move} within that round.
In case of ASYNC, there is no global notion of time.
The Fully-synchronous (FSYNC) is a special case of SSYNC, in which all the robots are activated in each round.
The algorithms presented in this paper work in both the ASYNC and the
SSYNC setting. For the sake of generality, we present our proofs
in terms of the ASYNC model.

\subsection{Notation and Terminology}

 A configuration $C$ is a multiset containing all the robot positions in 2D.
 At any time $t$ the configuration (the mapping of robots in the plane)
 is denoted by $C_t$. The convex hull of configuration $C_t$ is denoted as $CH_t$.
\emph{Convergence} is achieved when the distance between any pair of robots is less than a predefined value $c$ (and subsequently does not violate this anymore).
Our multi-robot system is vulnerable to adversarial manipulation,
however the algorithms presented in this paper are self-stabilizing~\cite{dolev2000self}
 and robust to state manipulations.
Since the robots are oblivious, they only depend on the \emph{current state}: if the state is perturbed,  the algorithms are still able to converge in a self-stabilizing manner \cite{Gilbert2009}.

\section{Convergence Algorithms}\label{sec:algo}

We now present distributed robot convergence algorithms for both our models,
$\mathcal{LD}$ and $\mathcal{OLA}$.

\subsection{Convergence for $\mathcal{LD}$}

In this section we consider the convergence problem for the monoculus robots in
the $\mathcal{LD}$ model. Our claims hold for any $c\geq 2b$.
Algorithm~\ref{algo:convergelocality}
distinguishes between two cases: (1) If the robot only sees one other robot,
it infers that the current configuration must be a line (of 2 or more robots),
and that this robot must be on the border of this line;
in this case,  the boundary robots always move inside (usual step size $b$).
(2) Otherwise, a robot moves towards  any visible, non-local robot (distance
at least $c$),
for a $b$ distance (the step size).

Our proof unfolds in a number of lemmas followed by a theorem.
First, Lemma~\ref{lem:4bdistance} shows that it is impossible to have a pair of robots with distance larger than $2c$ in the converged situation. Lemma~\ref{lem:chsubset} shows that our algorithm ensures a monotonically decreasing convex hull size.  Lemma~\ref{lem:finitedecrement} then proves that the decrement in perimeter for each movement is greater than a constant (the convex hull decrement is strictly monotonic). Combining all the three lemmas, we obtain the correctness proof of the algorithm.
In the following, we call two robots \emph{neighboring} if they see each other (line of sight is not obstructed by another robot).

\begin{algorithm}
\DontPrintSemicolon
\caption{\textsc{ConvergeLocality}}\label{algo:convergelocality}
\SetKwInOut{Input}{Input}\SetKwInOut{Output}{Output}
\Input{Any arbitrary configuration}
\Output{All robots are inside a circle of radius $c$}
\eIf{only one robot is visible}{Move distance $b$ towards that robot}
{\eIf{there is at least one robot farther than $c$}{Move distance $b$ towards any one of the robots with distance more than $c$}{Do not move\tcp*{All neighbor robots are within a distance $c$}}}
\end{algorithm}

\begin{lemma}\label{lem:4bdistance}
If there exists a pair of robots at distance more than $2c$ in a non-linear configuration, then there exists a pair of neighboring robots at distance more than $c$.
\end{lemma}
\begin{proof}
Proof by contradiction. If there is a pair of robots with distance more than $2c$, then for them not to move, there are at least two robots on the line joining them positioned such that each pair has a distance less than $c$. Since the robots are non-transparent, the end robots cannot look beyond their neighbors to know that there is a robot at a distance more than $c$.
In Fig.~\ref{fig:4bdistance}, $r_1$ and $r_4$ are $2c$ apart. So $r_2$ and $r_3$ block the view. Since it is a non-linear configuration, say robot $r_5$ is not on the line joining $r_1$ and $r_4$. $l$ is the perpendicular bisector of $\overline{r_1r_4}$. If $r_5$ is on the left side, then it is more than $c$ distance away from $r_4$ and vice versa. If there is another robot on $\overline{r_4r_5}$, then consider that as the new robot in a non-linear position, and we can argue similarly. Hence there would at least be a single robot similar to $r_5$ in a non-linear configuration for which the distance is more than $c$.

\begin{figure}[H]
\centering
\includegraphics[height=0.2\linewidth]{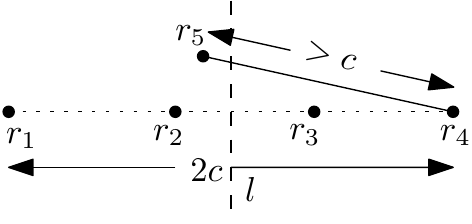}
\caption{A non-linear configuration with a pair of robots at a distance $2c$}\label{fig:4bdistance}
\end{figure}
\qed
\end{proof}

\begin{lemma}\label{lem:chsubset}
For any time $t' > t$ before convergence, $CH_{t'}\subseteq CH_{t}$.
\end{lemma}
\begin{proof}
The proof follows from a simple observation. Consider any robot $r_i$. If $r_i$ decides to move towards some robot, say $r_j$, then it is at least $c$ distance away. Even if $r_j$ is on the boundary, $r_i$ cannot cross the boundary. If $r_i$ is already on the boundary, then it always moves on the perimeter or inside the convex hull. Hence the convex hull gradually decreases.

 If all the robots are on a straight line, then the boundary robots move monotonically closer in each step. The distance between the end robots is a monotonically decreasing sequence until it reaches $c$.
\qed
\end{proof}

\begin{lemma}\label{lem:finitedecrement}
In finite time the decrement in the perimeter of the convex hull is at least $ b \left(1- \sqrt{\frac{1}{2}\left(1+\cos\left(\frac{2\pi}{n}\right)\right)}\right)$.
\end{lemma}
\begin{proof}
The sum of internal angles of a $k$-sided convex polygon is $(k-2)\pi$. So there exists a robot $r$ at a corner $A$ (ref. Fig.~\ref{fig:decrement}) of the convex hull such that the internal angle is less than $(1-\frac{2}{n})\pi$, where $n$ is the total number of robots.
Let $B$ and $C$ be the points where the  circle centered at $A$ with radius $b/2$ intersects the convex hull.
Any robot lying outside the circle will not move inside the circle according to Algorithm~\ref{algo:convergelocality}.
All the robots inside the circle will eventually move out once they are activated.
After all the robots are activated at least once, the decrement in perimeter is at least $AB + AC - BC$. From cosine rule,
$$AB + AC - BC = \frac{b}{2}+\frac{b}{2} -\sqrt{\left(\frac{b}{2}\right)^2+\left(\frac{b}{2}\right)^2-  2\frac{b}{2}\frac{b}{2}\cos\left(\pi -\frac{2\pi}{n}\right)}  $$
$$ = b \left(1- \sqrt{\frac{1}{2}\left(1+\cos\left(\frac{2\pi}{n}\right)\right)}\right)$$
\begin{figure}[H]\centering
\includegraphics[height=0.3\linewidth]{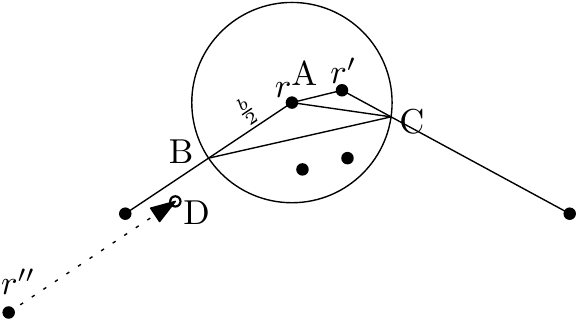}
\caption{Once activated the robots $r$ and $r'$ will move outside the solid circle with radius $b/2$. The robot $r''$ moves a distance $b$ towards $r'$ because distance between them is more than $2b$ and stops at $D$.}\label{fig:decrement}
\end{figure}
\qed
\end{proof}
%

\begin{theorem}(Correctness)
Algorithm~\ref{algo:convergelocality} terminates when all the robots are within a $c$ radius disc.
\end{theorem}
\begin{proof}
From Lemmas~\ref{lem:chsubset}
and~\ref{lem:finitedecrement} we know that the convex hull never
increases, and eventually a robot will be activated which strictly
decreases the hull.
According to Lemma~\ref{lem:4bdistance}, eventually there will not be a pair of robots with more than $2c$ distance.
Note that the distance between any two points in a disc of radius $c$ is less than or equal to $2c$. Hence the robots will converge within a disc of radius $c$.
\qed\end{proof}

\subsection{Convergence for $\mathcal{OLA}$}

In this section we consider monoculus robots in the $\mathcal{OLA}$ model.
Our algorithm will distinguish between
\textit{boundary-}, \textit{corner-} and \textit{inner-robots}, defined
in the canonical way. We note that robots can determine their type:
From the Fig.~\ref{fig:orthogonalline}, we can observe that for $r_2$, all the robots lie below the horizontal line. That means, one side of the horizontal line is empty
and therefore $r_2$ can figure out that it is a boundary robot. Similarly all $r_i$, $i\in \{3,4,5,6,7,8\}$ are boundary robots. Whereas, for $r_1$, both horizontal and vertical lines have one of the sides empty, hence $r_1$ is a corner robot. Other robots are all inner robots.
 Consequently, we define \textit{boundary robots} to be those, which have exactly one side of one of the orthogonal lines empty.

  Algorithm~\ref{algo:convergequadrant} (\textsc{ConvergeQuadrant}) can be described as follows.
  A rectangle can be constructed with lines parallel to the orthogonal lines passing through boundary robots such that, all the robots are inside this rectangle.
  In Fig.~\ref{fig:orthogonalline}, each boundary robot always moves inside the rectangle perpendicular to the boundary and the inside robots do not move.
  Note that the corner robot $r_1$ has two possible directions to move. So it moves toward any robot in that common quadrant.
  Gradually the distance between opposite boundaries becomes smaller and smaller and the robots converge. In case all the robots are on a line which is parallel to either of the orthogonal lines, then the robots will find that both sides of the line are empty.
  In that case they should not move. But the robots on either end of the line would only see one robot. So they would move along the line towards that robot.

\begin{figure}[!h]\centering
\includegraphics[height=0.4\linewidth]{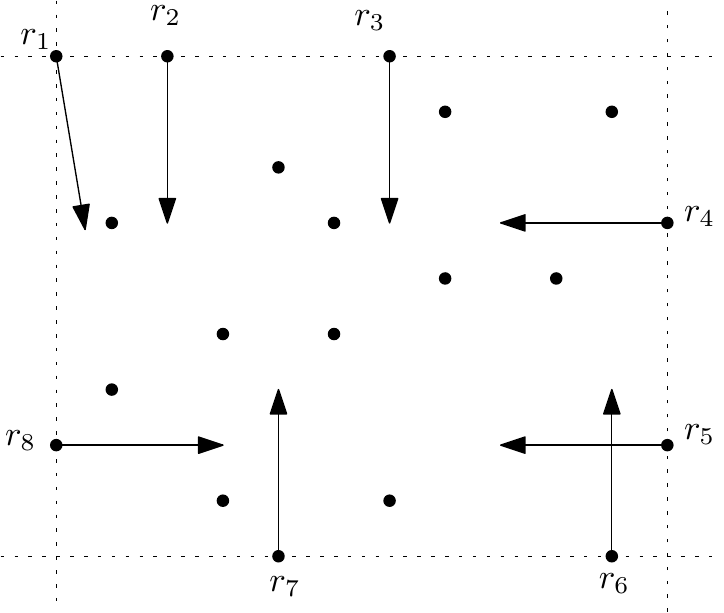}
\caption{Movement direction of the boundary robots}\label{fig:orthogonalline}
\end{figure}

\begin{algorithm}[!h]
\DontPrintSemicolon
\caption{\textsc{ConvergeQuadrant}}\label{algo:convergequadrant}
\SetKwInOut{Input}{Input}\SetKwInOut{Output}{Output}
\Input{Any arbitrary configuration and robot $r$}
\Output{All robots are inside a square with side $2b$ }
\uIf{only one robot is visible}{Move towards that robot}
\uElseIf{$r$ is a boundary robot}{Move perpendicular to the boundary to the side with robots}
\uElseIf{$r$ is a corner robot}{Move towards any robot in the non-empty quadrant}
\Else{Do not move \tcp*{It is an inside robot}}
\end{algorithm}

\begin{theorem}(Correctness)
Algorithm~\ref{algo:convergequadrant} moves all the robots inside some $2b$-sided square in finite time.
\end{theorem}
\begin{proof}
Consider the distance between the robots on the left and right boundary. The horizontal distance between them decreases each time either of them gets activated. The rightmost robot will move towards the left and the leftmost will move towards the right. The internal robots do not move. So in at most $n$ activation rounds of the boundary robot, the distance between two of the boundary nodes will decrease by at least $b$. Hence the distance is monotonically decreasing until $2b$.
Afterwards, the total distance will never exceed $2b$ anymore.

Given there is a corner robot present in the configuration, that robot will move towards any robot in the non-empty quadrant. So, the movement of
the corner robot contributes to the decrement in distance in both directions.
Consider robots inside the quadrant are presently very close to one of the boundaries and the corner robot moves towards that robot, then the decrement in one of the dimensions can be small (an $\epsilon > 0$).
Consider for example the configuration of a strip of width $b$,
then the corner robot becomes the adjacent corner in the next round;
this can happen only finitely many times.
Each dimension converges within a distance $2b$, so in the converged state the shape of the converged area would be $2b$-sided square.
\qed\end{proof}

\begin{remark}
If the robots have some sense of angular knowledge, the corner robots can always move in a $\pi/4$ angle, so the decrement in both dimension is significant, hence convergence time is less on average.
\end{remark}

\section{Impossibility and Optimality}\label{sec:impossibilty}

Given these positive results, we now show that we cannot make the monoculus robots much weaker, otherwise we lose convergeability.

\begin{theorem}
There is no deterministic convergence algorithm for monoculus robots without any additional capability.
\end{theorem}
\begin{proof}
We prove the theorem using a symmetry argument.
Consider the two configurations $C_1$ and $C_2$ in Fig.~\ref{fig:indistinguishableconfig}.
In $C_1$, all the robots are equidistant from robot $r$, while in $C_2$, the robots are at different distances, however the relative angle of the robots is the same at $r$. Now considering the local view of robot $r$, it cannot distinguish between $C_1$ and $C_2$.
 Say a deterministic algorithm $\phi$ decides a direction of movement for robot $r$ in configuration $C_1$.
 Since both $C_1$ and $C_2$ are the same from robot $r$'s perspective, the deterministic algorithm outputs the same direction of movement for both cases.
 \begin{figure}[H]
\centering
\includegraphics[height=0.3\linewidth]{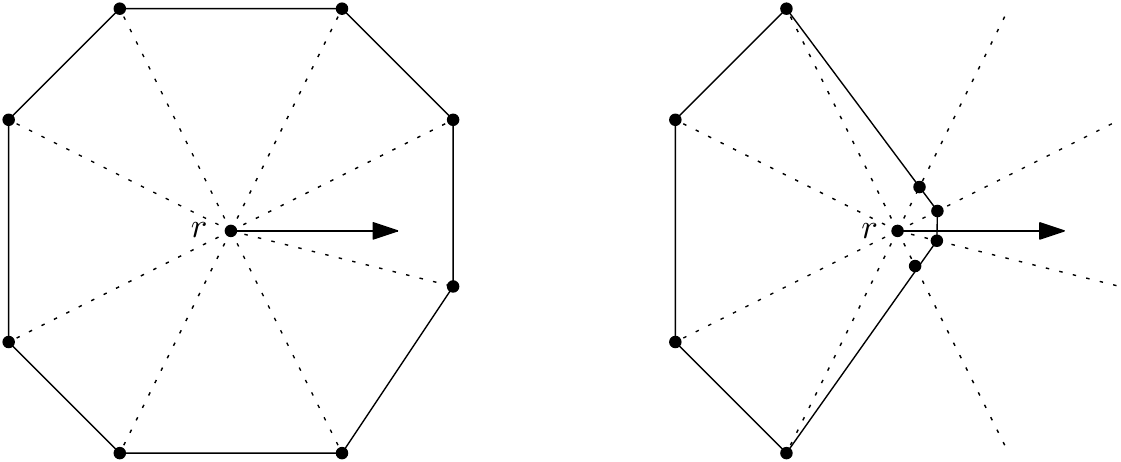}
\caption{Locally indistinguishable configurations with respect to $r$}\label{fig:indistinguishableconfig}
\end{figure}
Now consider the convex hull $CH_1$ and $CH_2$ of $C_1$ and $C_2$ respectively.
The robot $r$ moves a distance $b$ in one round.
The distance from any point inside $CH_1$ is more than $b$ but we can skew the convex hull in the direction of movement, so to make it like $CH_2$, where if the robot $r$ moves a distance $b$ it exits $CH_2$.
 Therefore there always exists a situation for any algorithm $\phi$ such that the area of the convex hull increases. Hence it is impossible for the robots to converge.

\qed\end{proof}

\section{Simulation}\label{sec:simul}

We now complement our formal analysis with simulations, studying the average
case. We assume that robots are distributed uniformly at random
in a square initially, that $b=1$
and $c=2$, and
we consider FSYNC scheduling. As a baseline to evaluate performance, we consider the
optimal convergence distance and time if the robots had capability to observe positions,
i.e., they are \emph{not} monoculus. Moreover, as a lower bound,
we compare to an algorithm which converges all robots to
the centroid, defined as follows:
$$ \{\bar{x}, \bar{y}\} = \left\{\cfrac{\sum_{i=1}^nx_i}{n},\cfrac{\sum_{i=1}^ny_i}{n}\right\}$$
\noindent where $\{x_i,y_i\} \forall i \in \{1,2,\cdots,n\}$ are the robots' coordinates.

We calculate distance $d_i$ from each robot to the centroid
in the initial configuration. The optimal distance we have used as
convergence distance is the sum of distances from each robot to the unit disc
centered at the centroid. So the sum of the optimal convergence distances $d_{opt}$ is given by
$$d_{opt} = \sum_{i=1}^n (d_i -1), \quad if \, d_i>1 $$
In the simulation of Algorithm~\ref{algo:convergelocality}, we define $d_{CL}$ as the cumulative number of steps taken by all the robots to converge (sometimes also known as the \emph{work}). Now we define the performance ratio, $\rho_{CL}$ as
 $$ \rho_{CL} = \cfrac{d_{CL}}{d_{opt}}$$
Similarly for Algorithm~\ref{algo:convergequadrant} we define $d_{CQ}$ and $\rho_{CQ}$. \\
 In Fig.~\ref{fig:varrobot}, we plot the distribution of 100 iterations of simulation of Algorithm~\ref{algo:convergelocality}, varying the number of robots for a fixed region of deployment.
The median increases if we increase the number of robots deployed in the same region.

In Fig.~\ref{fig:varrobotCQ}, we plot the distribution of 100 iterations of simulation of Algorithm~\ref{algo:convergequadrant} varying the number of robots for a fixed region of deployment.
The median increases if we increase the number of robots deployed in the same region.
In Fig.~\ref{fig:varregionCQ}, we plot the distribution of 100 iterations of simulation of Algorithm~\ref{algo:convergequadrant} for a fixed number of robots deployed in different regions. Here we can observe that the distribution does not vary much even if we change the region of deployment.
\begin{minipage}{\linewidth}
\begin{minipage}{0.45\linewidth}
\begin{figure}[H]
\includegraphics[width=\linewidth]{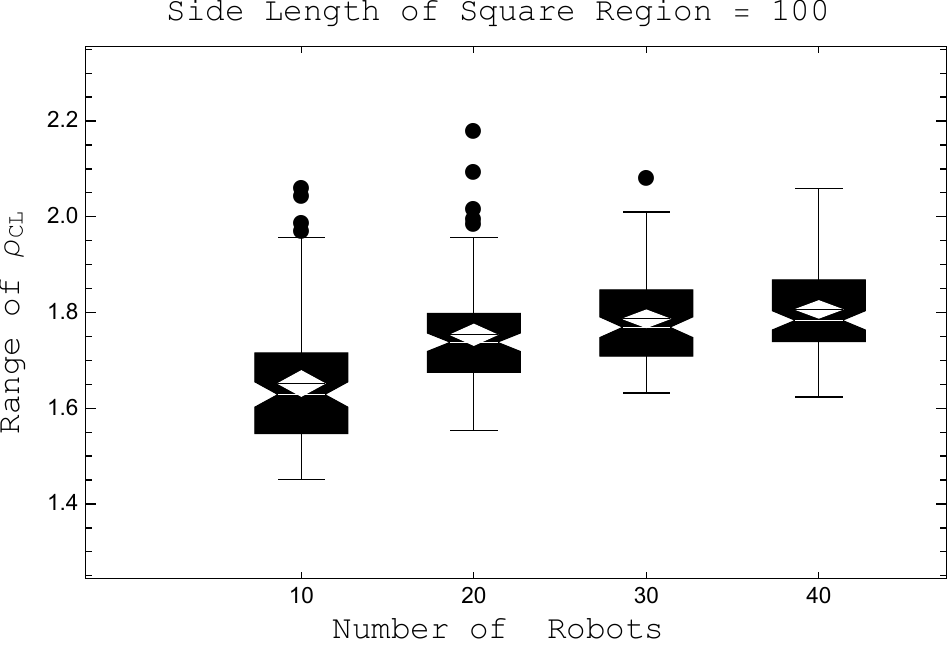}
\caption{Different number of robots in same region in $\mathcal{LD}$}\label{fig:varrobot}
\end{figure}
\end{minipage}\hspace{0.1\linewidth}
\begin{minipage}{0.45\linewidth}
\begin{figure}[H]
\includegraphics[width=\linewidth]{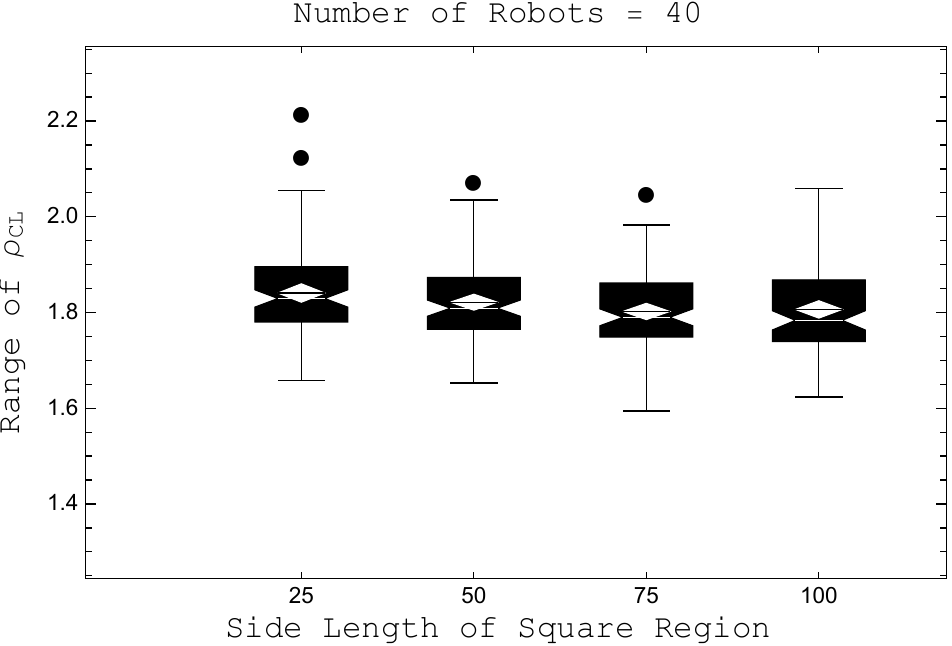}
\caption{Fixed number of robots deployed in different region in $\mathcal{LD}$}\label{fig:varrange}
\end{figure}
\end{minipage}
\begin{minipage}{0.45\linewidth}
\begin{figure}[H]
\includegraphics[width=\linewidth]{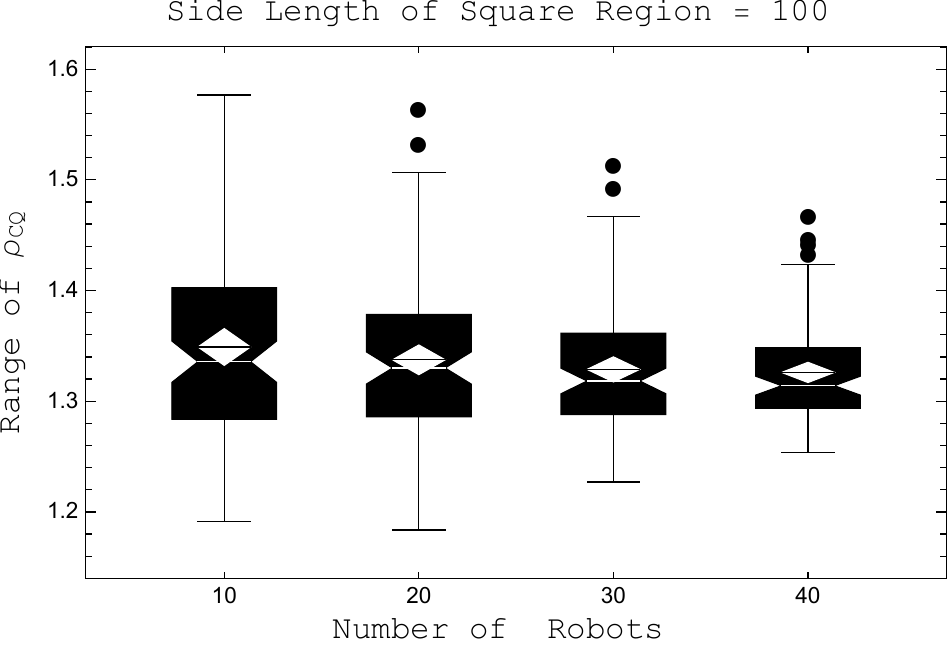}
\caption{Different number of robots in same region in $\mathcal{OLA}$}\label{fig:varrobotCQ}
\end{figure}
\end{minipage}\hspace{0.1\linewidth}
\begin{minipage}{0.45\linewidth}
\begin{figure}[H]
\includegraphics[width=\linewidth]{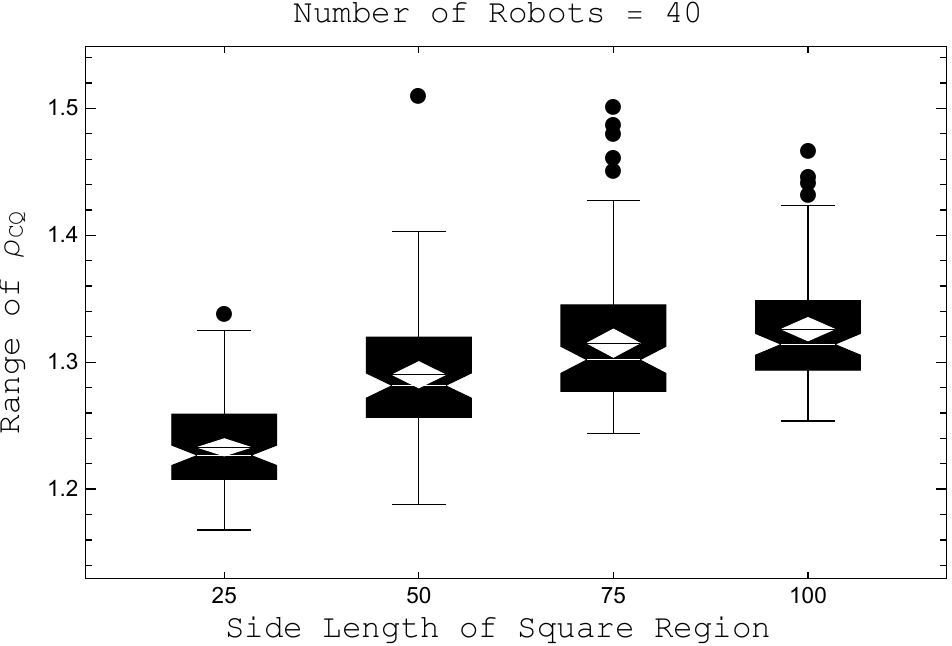}
\caption{Fixed number of robots deployed in different region in $\mathcal{OLA}$}\label{fig:varregionCQ}
\end{figure}
\end{minipage}
\end{minipage}\vspace{5mm}
In Fig.~\ref{fig:varrange}, we plot the distribution of 100 iterations of simulation of Algorithm ~\ref{algo:convergelocality} for a fixed number of robots deployed in different regions. Here we can observe that the distribution does not vary much even if we change the region of deployment.
\begin{figure}[!h]\centering
\includegraphics[height=0.37\linewidth]{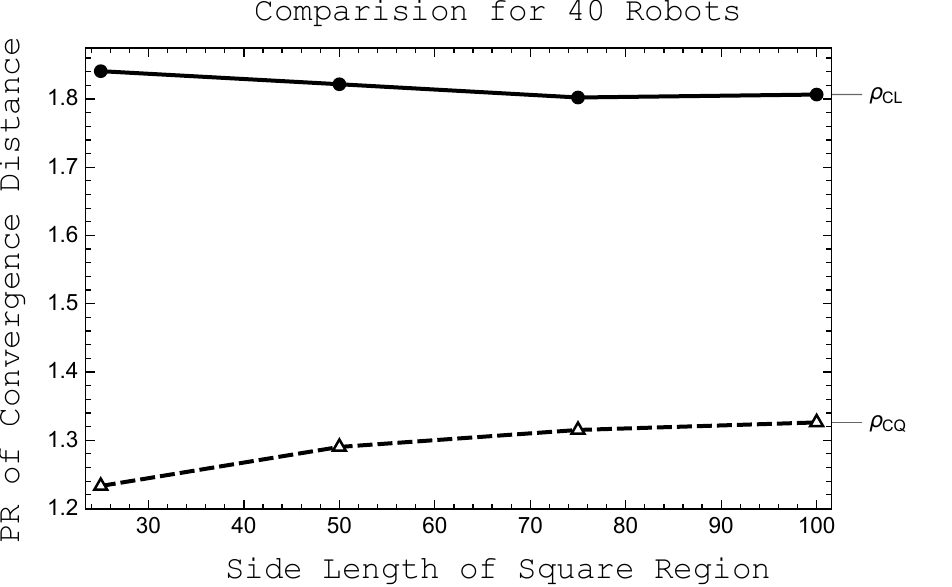}
\caption{ $\rho_{CL}$ VS $\rho_{CQ}$ for the same number of robots}\label{fig:compareCLCQrobotdistance}
\end{figure}
\begin{figure}[!h]\centering
\includegraphics[height=0.37\linewidth]{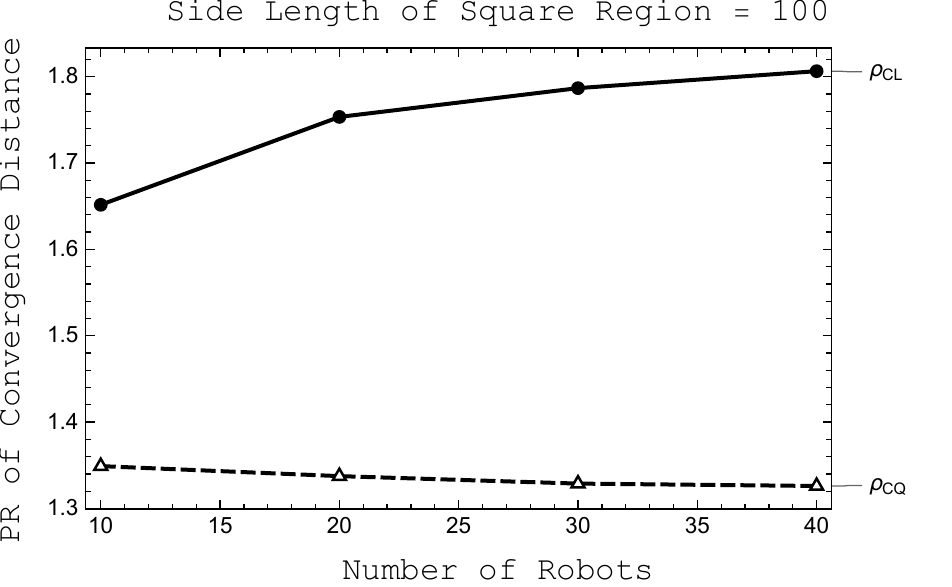}
\caption{ $\rho_{CL}$ VS $\rho_{CQ}$ for the same size of region}\label{fig:compareCLCQregiondistance}
\end{figure}
\begin{figure}[!h]\centering
\includegraphics[height=0.37\linewidth]{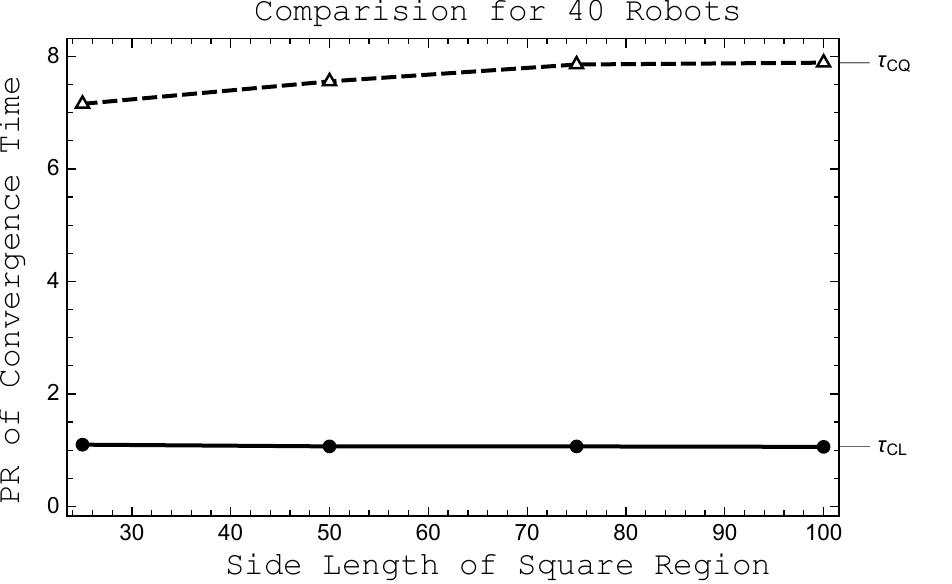}
\caption{ $\tau_{CL}$ VS $\tau_{CQ}$ for the same number of robots}\label{fig:compareCLCQrobot}
\end{figure}
\begin{figure}[!h]\centering
\includegraphics[height=0.37\linewidth]{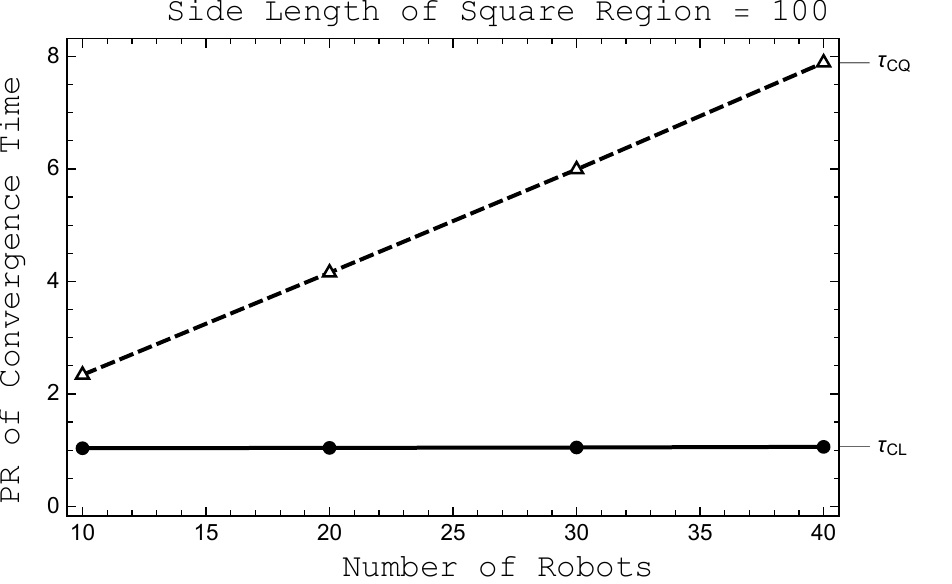}
\caption{$\tau_{CL}$ VS $\tau_{CQ}$ for the same size of region}\label{fig:compareCLCQregion}
\end{figure}

Fig.~\ref{fig:compareCLCQrobotdistance}, and ~\ref{fig:compareCLCQregiondistance} show the comparison between the performance ratio (PR) for distance. We can observe that Algorithm~\ref{algo:convergequadrant} performs better. This is due to the fact that, in Algorithm~\ref{algo:convergequadrant} only boundary robots move.

 Let $d_{max}$ be the distance of farthest robot from the centroid and  $t_{CL} $ be the number of synchronous rounds taken by Algorithm~\ref{algo:convergelocality} for convergence. We define $\tau_{CL}$ as follows
$$\tau_{CL} =\cfrac{t_{CL}}{d_{max}}$$
Similarly for Algorithm~\ref{algo:convergequadrant}, we define $t_{CQ}$ and $\tau_{CQ}$. $\tau_{CL}$ and $\tau_{CQ}$ show performance ratio for convergence time of Algorithm~\ref{algo:convergelocality} and ~\ref{algo:convergequadrant} respectively.
 In Fig.~\ref{fig:compareCLCQrobot} and ~\ref{fig:compareCLCQregion}, we can observe that $\tau_{CL}$ is very close to 1, so Algorithm~\ref{algo:convergelocality} converges in almost the same number of synchronous rounds (proportional to distance covered, since step size $b=1$) as the maximum distance.
We can observer that Algorithm~\ref{algo:convergequadrant} takes more time as the number of robots and the side length of square region increases.
\section{Discussion}\label{sec:discussion}

This section shows that our approach supports some interesting
extensions.

\subsection{Termination for $\mathcal{OLA}$ Model}

While we only focused on convergence and not termination so far,
we can show that with a small amount of memory,
termination is also possible in the $\mathcal{OLA}$
model.
To see this, assume that each robot has a 2-bit persistent memory in the $\mathcal{OLA}$ model for each dimension, total 4-bits for two dimensions. 
 Algorithm~\ref{algo:convergequadrant} has been modified to Algorithm~\ref{algo:convergequadranttermination} such that it can accommodate termination.
All the bits are initially set to 0. 
Each robot has its local coordinate system, which remains consistent over the execution of the algorithm. The four bits correspond to four boundaries in two dimensions, i.e., left, right, top and bottom.
 If a robot finds itself on one of the boundaries according to its local coordinate system, then it sets the corresponding bit of that boundary to 1. Once both bits corresponding to a dimension become 1, the robot stops moving in that dimension.
Consider a robot $r$. Initially it was on the left boundary in its local coordinate system. Then it sets the first bit of the pair of bits corresponding to $x$-axis. It moves towards right. Once it reaches the right boundary, then it sets the second bit corresponding to $x$-axis to 1. Once both the bits are set to 1, it stops moving along the $x$-axis. Similar movement termination happens on the $y$-axis also. Once all the 4-bits are set to 1, the robot stops moving.

%
\begin{algorithm}[!h]
\caption{\textsc{ConvergeQuadrantTermination}}\label{algo:convergequadranttermination}
\SetKwInOut{Input}{Input}\SetKwInOut{Output}{Output}
\DontPrintSemicolon
\Input{Any arbitrary configuration and robot $r$ with 4-bit memory}
\Output{All robots are inside a square with side $2b$ }
\eIf{the robot is on a boundary(ies)}{set the corresponding bit(s) to 1}{Do nothing \tcp*{$r$ is an inside robot}}
\uIf{$r$ is a boundary robot and the bits corresponding to that dimension are not 1}{Move perpendicular to the boundary to the side with robots}
\uElseIf{$r$ is a corner robot}{\eIf{Both bits corresponding to a dimension is 1}{Move in other dimension to the side with robots}{Move towards any robot in the non-empty quadrant}}
\Else{Do not move \tcp*{$r$ is not on boundary OR all four bits are 1}}
\end{algorithm}


\subsection{Extension to $d$-Dimensions}

Both our algorithms can easily be extended to $d$-dimensions.
For the $\mathcal{LD}$ model, the algorithm remains exactly the same.
For the proof of convergence, similar arguments as Lemma~\ref{lem:finitedecrement} can be used in $d$ dimensions.
We can consider the convex hull in $d$-dimensions and the boundary robots of the convex hull always move inside. The size of convex hull reduces gradually and the robots converge.

Analogously for the $\mathcal{OLA}$ model, the distance between two robots in the boundary of any dimension gradually decreases and the corner robots always move inside the $d$-dimensional cuboid. Hence it converges.
Here the robot would require $2d$ number of bits for termination.

\section{Related Work}\label{sec:relwork}

 The problems of gathering \cite{suzuki1999distributed}, where all the robots gather at a single point,
  convergence \cite{cohen2006convergence}, where robots come very close to each other and
 Pattern formation \cite{flochini1,suzuki1999distributed} have been studied intensively in the literature.

   Flocchini et al. \cite{FlocchiniPSW99} introduced the CORDA or Asynchronous (ASYNC) scheduling model for weak robots. Suzuki et al. \cite{suzu} have introduced the ATOM or Semi-synchronous (SSYNC) model.
       In \cite{suzuki1999distributed}, impossibility of gathering for $n=2$ without assumptions on local coordinate system agreement for \textit{SSYNC} and \textit{ASYNC} is proved.
Also, for $n>2$ it is impossible to solve gathering without assumptions on either coordinate system agreement or multiplicity detection \cite{Prencipe2007}. Cohen and Peleg \cite{CohenP04} have proposed a center of gravity algorithm for convergence of two robots in ASYNC and any number of robots in SSYNC.

To the best of our knowledge in all the previous works, the mathematical models always assume that the robots can find out the location of other robots in their local coordinate system in the Look step. This in turn implies that the robots can measure the distance between any pair of robots albeit in their local coordinates. All the algorithms exploit this location information to create an invariant point or a robot where all the other robots gather. But in this paper we deprive the robots of the capability to determine the location of other robots. This leads to robots incapable of finding any kind of distance or angles.

Any kind of pattern formation requires these robots to move to a particular point of the pattern. Since the monoculus robots cannot figure out locations, they cannot stop at a particular point. Hence any kind  of pattern formation algorithm described in the previous works which requires location information as input are obsolete. Gathering problem is nothing but the point formation problem \cite{suzuki1999distributed}. Hence gathering is also not possible for the monoculus robots.


\section{Conclusion}\label{sec:conclusion}

This paper introduced the notion of \emph{monoculus robots}
which cannot measure distance: a practically relevant generalization of
existing robot models. We have proved that the two basic models
still allow for convergence (and with a small memory,
even termination), but with less capabilities,
this becomes impossible.

The $\mathcal{LD}$ model converges in an almost optimal number of rounds, while the $\mathcal{OLA}$ model takes more time.
But the cumulative number of steps is less for the $\mathcal{OLA}$ model compared to the $\mathcal{LD}$ model since only boundary robots move. Although we found in our simulations that the median and angle bisector strategies successfully converge, finding a proof accordingly remains  an
open question.
We see our work as a first step, and believe that the study of weaker robots opens an interesting field for future research.

\bibliographystyle{plain}
\bibliography{bib}
\end{document}